\documentclass[letterpaper, 10 pt, conference]{ieeeconf}  

\IEEEoverridecommandlockouts                              
\usepackage[T1]{fontenc}
\usepackage[utf8]{inputenc}
\usepackage[english]{babel}

\usepackage{fancybox,shadow,xcolor} 
\usepackage{setspace}
\usepackage{interval}
\usepackage{xcolor}
\usepackage{psfrag}

\usepackage{graphicx}
\usepackage{amssymb,amsmath}

\usepackage{amsthm}

\usepackage{mathrsfs}
\usepackage{mathdots}
\usepackage{algorithmic,algorithm}

\newtheorem{thm}{Theorem}
\newtheorem{lem}{Lemma}
\newtheorem{prop}{Proposition}

\theoremstyle{definition}
\newtheorem{definition}{Definition}
\newtheorem{assumption}{Assumption}
\newtheorem{rem}{Remark}

\usepackage{interval}

\DeclareMathOperator*{\sign}{sign}
\DeclareMathOperator*{\diag}{diag}

\DeclareMathOperator*{\vvec}{vec}

\renewcommand{\Re}{\mathbb{R}}

\usepackage{bbm}

\newcommand{\BM}{\begin{bmatrix}}
\newcommand{\EM}{\end{bmatrix}}

\newcommand{\BBM}{\big[\begin{matrix}}
\newcommand{\EEM}{\end{matrix}\big]}

\newcommand{\bbm}{[\begin{matrix}}
\newcommand{\eem}{\end{matrix}]}

\title{\LARGE \bf On the tightest interval-valued state estimator for linear systems }

\author{Laurent Bako$^1$ and Vincent Andrieu$^2$
\thanks{$^1$L. Bako is with Universit\'{e} de Lyon, Amp\`{e}re (Ecole Centrale Lyon, INSA Lyon, Universit\'{e} Claude Bernard, CNRS UMR 5005), France, F-69611.
        {\tt\small E-mail: laurent.bako@ec-lyon.fr}}%
\thanks{$^2$V. Andrieu  is with Universit\'{e} de Lyon -- CNRS UMR
5007, LAGEP, France.
        {\tt\small E-mail: vincent.andrieu@gmail.com}}%
}

\begin{document}
\setstretch{.99}

\maketitle
\thispagestyle{empty}
\pagestyle{empty}

\begin{abstract}
This paper discusses an interval-valued state estimator for linear dynamic systems. In particular, we derive an expression of the tightest possible interval estimator in the sense that it is the intersection of all interval-valued estimators. This estimator appears, in a general setting, to be an infinite dimensional dynamic system.  Therefore practical implementation requires some over-approximations which would yield a good trade-off between computational complexity and tightness. 
\end{abstract}

\section{Introduction}
State estimation is a key problem in control engineering and more generally, in decision-making systems. 
One approach to  this estimation problem is that of interval observers. 
Contrary to classical observers which generate  single valued state estimates 
\cite{OReilly12-Book,Andrieu06-SIAM}, interval observers form a class of robust observers which produce set-valued estimates of the state for uncertain dynamical systems. The philosophy of the approach is inspired by the so-called set-membership estimation framework  \cite{Jaulin01-Book,Chisci96-Automatica}.
However, as stated in  \cite{Mazenc14-IJNRC},  the concept of interval observer can be traced back to \cite{Gouze00-EM}. 
More precisely,  interval observers are  dynamical systems  which  provide an interval-valued trajectory  (defined by a pair of lower and upper bounds) which contains all possible state trajectories of a given uncertain system.
In this approach model uncertainties arise from input disturbances, sensor noises and unknown initial conditions.
Assuming these signals take values  in known (bounded) time-dependent  intervals, the goal is to find an interval containing the state trajectories. 
In these settings, many contributions have been made for different classes of systems: continuous-time Linear Time Invariant (LTI) \cite{Mazenc11-Automatica,Cacace15-TAC,Combastel13-TAC,Meslem16-CoDIT}, discrete-time LTI/LTV systems \cite{Efimov13-TAC,Mazenc14-IJNRC,Mazenc13-Automatica}, Linear Parameter Varying (LPV) systems \cite{Chebotarev15-Automatica,Efimov13-TAC-b}, nonlinear systems \cite{Raissi12-TAC,Moisan10-SCL}. For more on the interval observer literature we refer to a recent survey reported in \cite{Efimov16-ARC}.

Although the existing literature covers a large variety of systems, a question of major importance that has not received much attention so far is that of the size (or volume) of the estimated interval set. In effect, there exist in principle infinitely many interval estimators that satisfy the outer-bounding condition for the state trajectories of the system of interest.  But ideally, one would like to find the smallest possible interval set  (in some sense)  among all those which enclose the states. Hence we ask the question of how to characterize the tightest interval estimator in the sense that the upper and lower bounds are closest componentwise.

The current paper intends  to address this question.  For this purpose, a new approach is proposed to tackle the problem of designing interval estimator. 
%
For simplicity of exposition we restrict our attention to continuous-time LTI systems but the proposed approach can be extended at a moderate effort to  Linear Time-Varying (LTV) systems. The key ingredient of the proposed framework is a parametrization of the interval set in the form of a center jointly with a radius which measures the width (size) of the interval set. Then we show that simple maximization techniques allow to construct the tightest enclosing interval set. Note however that, computing numerically this tight interval-valued estimate of the state is expensive in general. We therefore discuss some approximation strategies illustrating the trade-off between quality (tightness indeed) of the estimate and the computational price to pay for it. 
Note that another aspect of the quest for tightness in interval estimation was discussed in \cite{Rami08-CDC}. There,  however, the problem was different from the one of the current paper; it was about finding an observer gain to minimize an $\ell_1$ norm  of the width of the interval estimator.

\textbf{Outline.} The remainder of the paper is organized as follows. In Section \ref{sec:Preliminaries}, we set up the estimation problem and present the technical material  employed 
for designing the estimator. 
In Section \ref{sec:Open-Loop} we discuss estimators in open-loop, that is, estimators that result only from the simulation of the state transition equation without any use of the measurement. Section \ref{sec:Closed-Loop} discusses a systematic way of transforming a classical observer into an interval-valued estimator. Section \ref{sec:Simulations} reports some numerical results confirming tightness of the proposed estimator. We conclude the paper in Section \ref{sec:Conclusion}.  

\textbf{Notations.}
$\Re$ (resp. $\Re_+$) is the set of real (resp. nonnegative real) numbers. For  a real  number  $x$, $|x|$ will refer to the absolute value  of $x$.  
For $x=\bbm x_1 & \cdots & x_n\eem^\top\in \Re^n$, $\left\|x\right\|_p$ will denote the $p$-norm of $x$ defined by $\left\|x\right\|_p=(|x_1|^p+\cdots+|x_n|^p)^{1/p}$, for $p\geq 1$. In particular for $p=\infty$, $\left\|x\right\|_\infty=\max_{i=1,\ldots,n}\left|x_i\right|$. 
For a matrix $A\in \Re^{n\times m}$, $\left\|A\right\|_2$ will denote the $2$-norm of the vector $\vvec(A)$ obtained by vectorizing the matrix $A$, i.e., $\left\|A\right\|_2=\left\|\vvec(A)\right\|_2$ (this norm is also called the Frobenius norm of the matrix $A$ and denoted $\left\|A\right\|_F$). 
If $A=[a_{ij}]\in \Re^{n\times m}$ and $B=[b_{ij}]\in \Re^{n\times m}$ are real matrices of the same dimensions, the notation $A\leq B$ will be understood as an elementwise inequality on the entries, i.e., $a_{ij}\leq b_{ij}$ for all $(i,j)$.  $|A|$ corresponds to the matrix $[|a_{ij}|]$ obtained by taking the absolute value of each entry of $A$. 
In case $A$ and $B$ are real square symmetric matrices, $A\succeq B$ (resp. $A\succ B$) means that $A-B$ is positive semi-definite (resp. positive definite). A square matrix $A$ is called \textit{Hurwitz} if all its eigenvalues have negative real parts. It is called \textit{Metzler} if all its off-diagonal entries are nonnegative. For a positive integer $n$,  we use the notation  $\mathcal{L}^{n}(\Re_+)=\left\{s:\Re_+\rightarrow\Re^n\right\}$ to refer to  the set of $n$ dimensional vector-valued  functions on $\Re_+$. 
$\mathcal L^\infty(\Re_+,\Re^{n})$ concerns the case in which the functions are bounded and measurable.
\vspace{-.15cm}
\section{Preliminaries}\label{sec:Preliminaries}
\subsection{Estimation problem settings}
\noindent Consider a Linear Time Invariant (LTI) system described by
\begin{equation}\label{eq:LTI}
	\begin{aligned}
		&\dot{x}(t)=Ax(t)+Bw(t) \\
		&y(t)= Cx(t)+v(t) 
	\end{aligned}
\end{equation}
where the state $x$ takes values in $\Re^n$,
 $w$ and $v$ are (possibly unknown) external signal respectively in $\mathcal L^\infty(\Re_+,\Re^{n_w})$ and $\mathcal L^\infty(\Re_+,\Re^{n_y})$,
$y\in \mathcal{L}^{n_y}(\Re_+)$ is a measured output. 
$A\in \Re^{n\times n}$, $B\in \Re^{n\times n_w}$ and  $C\in \Re^{n_y\times n}$ are some real matrices.

First of all, we define intervals of $\Re^n$. 
Let $\underline{x}$ and $\overline{x}$ be two vectors in $\Re^n$ such that $\underline{x}\leq \overline{x}$ with the inequality holding componentwise.  An interval of $\Re^n$, denoted $\interval{\underline{x}}{\overline{x}}$,  is the subset defined by
\begin{equation}\label{eq:Interval}
	\interval{\underline{x}}{\overline{x}}= \left\{x\in \Re^n: \underline{x}\leq x \leq \overline{x}\right\}. 
\end{equation}

\noindent Now we consider the following assumption. 
\begin{assumption}\label{assum:Bounding}
There exist (known) bounded signals $(\underline{w}, \overline{w})$ and $(\underline{v}, \overline{v})$ respectively in  $\mathcal L^\infty(\Re_+,\Re^{n_w})$ and in $\mathcal L^\infty(\Re_+,\Re^{n_y})$ such that
$\underline{w}(t)\leq w(t)\leq  \overline{w}(t) $ and $\underline{v}(t)\leq v(t)\leq  \overline{v}(t) $ for all $t\in \Re_+$.
Here again the inequalities are understood componentwise.
\end{assumption}

We consider in this paper the problem of  synthesizing 
an \textit{interval estimator} for the state of the LTI system \eqref{eq:LTI}. 
Considering that the initial state $x(0)$ of \eqref{eq:LTI} lives in an interval of the form $[\underline{x}(0), \overline{x}(0)]\subset\Re^{n}$ and that the external signals $w$ and $v$ satisfy Assumption \ref{assum:Bounding}, we want to estimate upper and lower bounds $\overline{x}(t)$ and $\underline{x}(t)$, $t\in \Re_+$, for all  possible state trajectories of the uncertain system \eqref{eq:LTI}. 
\begin{rem}
 A causal dynamical system $\Sigma$ with input $\xi \in \mathcal{L}^{n_\xi}(\Re_+)$, output $z\in \mathcal{L}^{n_z}(\Re_+)$ and initial state $X_0\in \Re^{n_x}$ can be described by a state-space realization (similar to  the one  in \eqref{eq:LTI}) or by its  input-output map $f_\Sigma:\Re_+ \times\Re_+\times \Re^{n_x}\times \mathcal{L}^{n_\xi}(\Re_+)$ defined by $z(t)=f_\Sigma(t_0,t,X_0,\xi)$, $t\geq t_0$. $z(t)$ is hence the output of the system  $\Sigma$ at time $t$ when it starts at time $t_0\leq t$ in state $X_0$ and is driven by the input $\xi$. See e.g., \cite{Willems71-SIAM} for more on this formalism. 
\end{rem}
\begin{definition}[Interval estimator]\label{def:Interval-Estimator}
Consider the system \eqref{eq:LTI} and pose $b_w(t)=\bbm \underline{w}(t)^\top & \overline{w}(t)^\top\eem^\top$, $b_v(t)=\bbm \underline{v}(t)^\top & \overline{v}(t)^\top\eem^\top$, $\xi(t)= \bbm b_w(t)^\top & b_v(t)^\top & y(t)^\top\eem^\top$ and $X_0=\bbm \underline{x}(t_0)^\top, \overline{x}(t_0)^\top\eem^\top$ for some $t_0$.
Consider a causal dynamical system with output $(\underline{x}, \overline{x})$ defined by its input-output maps $(F,G)$ as 
\begin{equation}\label{eq:interval-estimator}
\left\{	\begin{aligned}
		&\underline{x}(t)=F(t_0,t,X_0,\xi) \\
		&\overline{x}(t)=G(t_0,t,X_0,\xi),
	\end{aligned} \right. 
\end{equation}
where $F:\Re_+\times \Re_+\times \Re^{2n}\times\mathcal{L}^{n_\xi }(\Re_+)\rightarrow \Re^n$ and $G$ (defined similarly to $F$) are some operators.\\ 
The system \eqref{eq:interval-estimator} is called an \textit{interval estimator} for system \eqref{eq:LTI} if: 
\begin{enumerate}
	\item[(a)] Any state trajectory $x$ of \eqref{eq:LTI} satisfies   $\underline{x}(t)\leq x(t)\leq  \overline{x}(t) $  for all $t\geq t_0$, whenever $\underline{x}(t_0)\leq x(t_0)\leq  \overline{x}(t_0) $
	\item[(b)] \eqref{eq:interval-estimator} is Bounded Input-Bounded Output (BIBO) stable, i.e. $(\underline{x}, \overline{x})$ is bounded whenever $X_0$ and $\xi$ are bounded.  
\end{enumerate}
\end{definition}
\noindent Here  the signals $b_w$, $b_v$, $y$ and the initial state vector $X_0$ are viewed as the inputs of system \eqref{eq:interval-estimator}.  Boundedness is understood in the sense of the infinity norm being finite. 

We will discuss two types of interval estimators: open-loop interval estimators (or simulators) where \eqref{eq:interval-estimator} does not depend on the measurements $y$ and the measurement noise $b_v$ ; and closed-loop interval estimators where measurement is fed back to the estimator.  \\
There are in principle infinitely many estimators that qualify as interval estimators in the sense of  Definition \ref{def:Interval-Estimator}. It is therefore desirable to define a performance index (measuring e.g. the size of the estimator) which selects the best estimator among all. We will be interested here in the smallest interval estimator in the following sense. 
\vspace{-2pt}
\begin{definition}\label{def:smallest-interval}
Let $\mathcal{S}$ denote a subset of $\Re^n$. An interval $\mathcal{I}_{\mathcal{S}}\subset\Re^n$ is called the tightest interval containing $\mathcal{S}$ if $\mathcal{S}\subset \mathcal{I}_{\mathcal{S}}$ and if for any interval $\mathcal{J}$ of $\Re^n$, 
$\mathcal{S}\subset \mathcal{J} \: \Rightarrow \: \mathcal{I}_{\mathcal{S}} \subset \mathcal{J}.$
\end{definition}
\noindent In other words, the tightest interval $\mathcal{I}_{\mathcal{S}}$ "generated" by $\mathcal{S}\subset\Re^n$ is the intersection of all intervals containing $\mathcal{S}$. 
\vspace{-.15cm}
\subsection{Preliminary material on interval representation}\label{subsec:Interval-Representation} 
An important observation for future developments of the paper is that $\interval{\underline{x}}{\overline{x}}$ can be equivalently represented by a set of the form 
\begin{equation}
	\mathcal{C}(c_x,p_x)=\big\{c_x+P_x \alpha : \alpha\in \Re^n, \: \left\|\alpha\right\|_\infty \leq 1\big\}
\end{equation}
where 
 \begin{equation}
	 c_x=\dfrac{\overline{x}+\underline{x}}{2}, \quad P_x=\diag\big(p_x\big), \quad p_x = \dfrac{\overline{x}-\underline{x}}{2}
 \end{equation}
The notation $\diag(v)$ for a vector $v=\bbm v_1 & \cdots & v_n\eem^\top$ refers to  the diagonal matrix whose diagonal elements are the entries of $v$. 
 We will call the so-defined $c_x$ the center of the interval $\interval{\underline{x}}{\overline{x}}$ and $p_x$ its radius. To sum up, the interval set can be equivalently represented by the pairs $(\underline{x},\overline{x})\in \Re^n\times \Re^n$ and $(c_x,p_x)\in \Re^n\times \Re_+^n$  i.e.,  $\interval{\underline{x}}{\overline{x}}=\mathcal{C}(c_x,p_x)$.  
Finally, it will be useful to keep in mind that then $\underline{x}=c_x-p_x$ and $\overline{x}=c_x+p_x$. 
%

The following lemma states a key result for later uses. 
\begin{lem}\label{lem:Tightest-Interval}
Let $(c_z,p_z)\in \Re^{m}\times \Re_+^{m}$ and $(c_w(t),p_w(t))\in \Re^{n_w}\times \Re_+^{n_w}$ be center-radius representations of some intervals $\interval{\underline{z}}{\overline{z}}$ and $\interval{\underline{w}(t)}{\overline{w}(t)}$ where 
 $c_w$ in $\mathcal L^\infty(\Re_+,\Re^{n_w})$  and $p_w$ in $\mathcal L^\infty(\Re_+,\Re_+^{n_w})$. 
Let $F\in \Re^{n\times m}$ be a fixed value matrix and  
$H$ be a matrix function in  $\mathcal L^\infty(\Re_+,\Re^{n\times n_w})$.
Consider the set $\mathbb{I}$ defined by  
\begin{multline}
	\mathbb{I}=  
	\left\{Fz+\int_{t_0}^{t_1}H(\tau)w(\tau)d\tau :\right. \\
	\Big. z\in \interval{\underline{z}}{\overline{z}}, \:  w \text{ measurable},  w(\tau)\in \interval{\underline{w}(\tau)}{\overline{w}(\tau)} \Big\}
\end{multline}
with $\interval{t_0}{t_1}$ being some interval of $\Re_+$.  
Finally, consider the pair $(c,p)$ defined by: 
\begin{align}
	& c = F c_z+\int_{t_0}^{t_1}H(\tau)c_w(\tau)d\tau \label{eq:c} \\
	& p = \left|F\right|p_z+\int_{t_0}^{t_1}\left|H(\tau)\right|p_w(\tau)d\tau \label{eq:p}
\end{align}
Then, 
$\interval{c-p}{c+p}$ is the tightest interval set enclosing $\mathbb{I}$ in the sense of Definition \ref{def:smallest-interval}. 
\end{lem}
\begin{proof}
We first show that $\mathbb{I}\subset \interval{c-p}{c+p}$. Let $x\in \mathbb{I}$. 
Then $x$ can be written in the form
$$ x = Fz+\int_{t_0}^{t_1}H(\tau)w(\tau)d\tau, $$ 
where $z$ and $w$ obey the conditions in the definition of $\mathbb{I}$. 
As discussed in Section \ref{subsec:Interval-Representation} we can describe the uncertain vector $z$ and  uncertain signal $w$  by $z=c_z+P_z\alpha_z$ and $w(\tau)=c_w(\tau)+P_w(\tau)\alpha_w(\tau)$ respectively with $\alpha_z\in \Re^{n}$ and $\alpha_w(\tau)\in \Re^{n_w}$ such that $\left\|\alpha_z\right\|_\infty\leq 1$ and $\left\|\alpha_w(\tau)\right\|_\infty\leq 1$ for all $\tau$ and $P_z=\diag(p_z)$, $P_w(\tau)=\diag(p_w(\tau))$. 
It follows, by plugging these representations in the expression of $x$, that $x=c+\psi$ with $c$ expressed as in \eqref{eq:c} and 
$$\psi=FP_z\alpha_z+\int_{t_0}^{t_1} H(\tau) P_w(\tau)\alpha_w(\tau)d\tau. $$
 Now  consider the  vectors $r^+\in \Re^n$ and $r^-\in \Re^n$ defined  for all $i=1,\ldots,n$, by
$$\begin{aligned}
	r_{i}^+=&\max\Big\{\psi_i: \left\|\alpha_z\right\|_\infty\leq 1,\Big. 
	 \Big. \left\|\alpha_w(\tau)\right\|_\infty\leq 1, \: \tau\in \interval{t_0}{t_1}\Big\}\\
	r_{i}^-=&\min\Big\{\psi_i: \left\|\alpha_z\right\|_\infty\leq 1,\Big. 
	 \Big. \left\|\alpha_w(\tau)\right\|_\infty\leq 1, \: \tau\in \interval{t_0}{t_1}\Big\}
\end{aligned}$$
with $\left\|\cdot\right\|_\infty$ referring to the infinite norm of vectors. 
By denoting the $i$-th row of $F$ with $f_i^\top$ and that of $H(\tau)$ with $h_i^\top(\tau)$, we see that the maximizing values of the decision variables are $\alpha_z^*=\sign(f_i)$ and  $\alpha_w^*(\tau)=\sign(h_i(\tau))$, $t_0\leq \tau \leq t_1$ hence leading to 
$$ r_{i}^+=\big|f_i^\top\big| p_z+\int_{t_0}^{t_1} \big|h_i^\top(\tau)\big|  p_w(\tau)d\tau.$$
Here $\sign$ refers to the sign function operating componentwise. 
Hence $r_i^+$ is equal to the $i$-th entry of $p$  defined in \eqref{eq:p}  and so $r^+=p$. 
Similarly it can be seen that $r^-=-p$. By definition of $r^+$ and $r^-$, it is obvious that $c+r^-\leq x\leq  c+r^+$. Hence $\mathbb{I}\subset \interval{c-p}{c+p}$. \\
For clarity of the rest of the proof, we additionally observe that since $-p_i=\min_{x\in \mathbb{I}}(x-c)_i$ and $p_i=\max_{x\in \mathbb{I}}(x-c)_i$ are minimum and maximum values respectively, they are attainable for some elements $s^i$ and $\tilde{s}^i$ of $\mathbb{I}$, i.e., $(s^i)_i=c_i-p_i$ and $(\tilde{s}^i)_i=c_i+p_i$. 
\\
\textit{Proof of tightness.} We are now left with proving that $\interval{c-p}{c+p}$ is the tightest enclosing interval set for $\mathbb{I}$. For this purpose, consider another interval set $\interval{\underline{g}}{\overline{g}}$  such that $\mathbb{I}\subset \interval{\underline{g}}{\overline{g}}$.   Pose $p_g=(\overline{g}-\underline{g})/2$ and $c_g=(\overline{g}+\underline{g})/2$ and consider $x\in \mathbb{I}$.  Since $x$ lies in the intersection of $\interval{c-p}{c+p}$ and $\interval{\underline{g}}{\overline{g}}$,  there is $\alpha$ and $\alpha_g$ all with infinity norm less than $1$, such that 
$x=c+\diag(p)\alpha =c_g+\diag(p_g)\alpha_g,$
which translates componentwise into
\begin{equation}\label{eq:component-wise}
	\begin{aligned}
		-p_{g,i}-c_{i}+c_{g,i}&\leq p_{i}\alpha_{i} \leq p_{g,i}-c_{i}+c_{g,i}
	\end{aligned}
\end{equation}
for $i=1,\ldots,n$. 
From the first part of the proof, we know that for any $i=1,\ldots,n$, there exists $(s^i,\tilde{s}^i)\in \mathbb{I}^2$ such that $(s^i)_i=c_i-p_i$  and $(\tilde{s}^i)_i=c_i+p_i$. Since $x$ is an arbitrary element of $\mathbb{I}$, the inequalities \eqref{eq:component-wise} must hold for both particular instances $x=s^i$ and $x=\tilde{s}^i$. And for these values of $x\in \mathbb{I}$, $\alpha_i$ in \eqref{eq:component-wise} clearly takes the values $-1$ and $+1$ respectively. 
It follows that 
$$-p_{g,i}-c_{i}+c_{g,i}\leq \pm p_{i} \leq p_{g,i}-c_{i}+c_{g,i}$$
from which we see that 
$c_{g,i}-p_{g,i}\leq c_{i}-p_{i}$ and $c_{i}+p_{i}\leq c_{g,i}+p_{g,i}$. 
Hence $\interval{c-p}{c+p}\subset \interval{\underline{g}}{\overline{g}}$. 
This shows that $\interval{c-p}{c+p}$ is the tightest interval containing $\mathbb{I}$. 
\end{proof}

\section{Open-loop interval estimator for LTI systems}\label{sec:Open-Loop}

\subsection{Open-loop simulation: the best interval estimator} 
We first discuss a robust simulation of the state trajectory of  the LTI system in \eqref{eq:LTI} under uncertain perturbation $w$  and when the initial state $x(0)$  belongs to  a known  interval set. For this purpose we will assume that the  matrices $A$ and $B$ have fixed and known values.  
We use the notations $(c_w(t),p_w(t))\in \Re^{n_w}\times \Re_+^{n_w}$ and $(c_x(t),p_x(t))\in \Re^{n}\times \Re_+^{n}$, $t\in  \Re_+$, to denote the center-radius representations for the intervals $[\underline{w}(t),\overline{w}(t)]$ and $[\underline{x}(t),\overline{x}(t)]$ respectively. 

\begin{thm}\label{thm:Tightest-Estimator}
Let the initial conditions and the uncertain input sets of system \eqref{eq:LTI} be described respectively  by $(c_x(0),p_x(0))$ and  $(c_w(t),p_w(t))$ for $t\in \Re_+$. Assume that system \eqref{eq:LTI} is stable, i.e., $A$ is Hurwitz. Then the interval $\interval{\underline{x}}{\overline{x}}$ defined by 
\begin{equation}\label{eq:IntEst}
	\underline{x}(t)=c_x(t)-p_x(t) \quad \mbox{ and } \quad \overline{x}(t)=c_x(t)+p_x(t)
\end{equation}
 with 
\begin{eqnarray}
 c_x(t)&=&e^{At}c_x(0)+\int_{0}^t e^{A(t-\tau)}B c_w(\tau)d\tau \label{eq:cx-continuous-time}\\
	p_x(t)&=&\left|e^{At}\right|p_x(0)+\int_{0}^t\big|e^{A(t-\tau)}B\big|p_w(\tau)d\tau, \label{eq:px-continuous-time}
\end{eqnarray}
is the tightest interval-valued estimator for system \eqref{eq:LTI} in the sense of Definition \ref{def:smallest-interval}.
\end{thm}
\begin{proof}
That $\interval{\underline{x}}{\overline{x}}$ defined in \eqref{eq:IntEst} is the tightest interval-valued trajectory containing the trajectories of \eqref{eq:LTI} is a statement that follows directly from Lemma \ref{lem:Tightest-Interval}. 
It suffices  to note that the solution of \eqref{eq:LTI} takes the form  
$$
x(t)=e^{At}x(0)+\int_{0}^t e^{A(t-\tau)}B w(\tau)d\tau
$$
and apply the lemma. 
As to condition (b) of Definition \ref{def:Interval-Estimator} it is an immediate consequence of the stability assumption on system \eqref{eq:LTI}. 
\end{proof}

Framed differently, the theorem states that the interval estimator \eqref{eq:IntEst}-\eqref{eq:px-continuous-time} is the intersection of all enclosing intervals for the state trajectories generated by the uncertain system \eqref{eq:LTI}. Now the question we ask is how to compute the proposed estimates. Of course, a direct implementation of the equations \eqref{eq:IntEst}-\eqref{eq:px-continuous-time} might be overly expensive in finite-time and unfeasible when the time horizon considered for estimation goes to infinity. We will therefore be searching, when possible, for a finite dimensional state-space realization for the signals $c_x$ and $p_x$. To begin with, note that $c_x$ can be simply realized as $\dot{c}_x=Ac_x+Bc_w$. So the challenge is rather related to the realization of $p_x$. In the sequel, we discuss a few particular cases where  a finite dimensional realization exists.

\subsection{On the realization of the tightest estimator}
We start by observing that if $A$  is a Metzler matrix and if $B$ is either nonpositive or nonnegative, then $p_x$ in \eqref{eq:px-continuous-time} can be simply realized by  $\dot{p}_x(t)=Ap_x(t)+|B|p_w(t)$. This follows from the fact that $e^{At}$ is a nonnegative matrix for all $t\geq 0$ whenever $A$ is a Metzler matrix.
Consequently, one can drop the absolute value symbols in \eqref{eq:px-continuous-time} hence yielding the simple realization displayed above. 

A second remark concerns the scenario where $p_w$ is constant. In this latter case, a simple realization of $p_x$ can be obtained as stated in the following proposition.  
\begin{prop}\label{eq:px-for-constant-pw}
Assume that $p_w(t)=p_w(0)$ for all $t\in \Re_+$, i.e., $p_w$ is constant. Then the signal $p_x$ in \eqref{eq:px-continuous-time} can be realized as follows:
\begin{equation}\label{eq:realization-px}
	\left\{\begin{aligned}
		&\dot{M}(t)=AM(t), \quad M(0)=I_{n}\\
		&\dot{r}(t) = \left|M(t)B\right|p_w(0), \quad r(0)=0\\
		&p_x(t)=\left|M(t)\right|p_x(0)+r(t)
	\end{aligned}\right.
\end{equation}
with state $(M(t),r(t))\in \Re^{n\times n}\times \Re^{n}$ and $I_{n}$ being the identity matrix of order $n$. 
\end{prop}
The proof of the proposition follows by simple calculations. 
\noindent We will show below that even though $p_w$ is not constant in general, we can rely on this proposition to construct a nice over-approximation of the tightest interval estimator.  
\subsection{Some approximations of the tightest interval estimator}\label{subsec:Approximation}
As it turns out, apart from some special situations, implementing the tight estimator \eqref{eq:IntEst}-\eqref{eq:px-continuous-time} in the most general case is intractable in practice. We therefore consider in this section the question of whether one could over-approximate $p_x$ by a more easily realizable signal $\hat{p}_x$. 
 In order to discuss this question, let us recall  some basic facts  that will be useful.
\begin{lem}\label{lem:Inequalities}
Let $A$ and $B$ be matrices of compatible dimensions. Then the following properties hold (with all matrix inequalities understood entrywise): 
\begin{subequations}
\begin{align}
	&\left|A+B\right|\leq \left|A\right|+\left|B\right| \label{eq:sum}\\
	&\left|AB\right|\leq \left|A\right|\left|B\right|  \label{eq:product}\\
	&|A|\leq B \quad \Rightarrow \quad  \left\|A\right\|_2\leq \left\|B\right\|_2 \label{eq:inequality-norm} \\
	&\left\|A\right\|_2=\left\|\left|A\right|\right\|_2 \label{eq:equality-norm}\\
	& \left|e^{A}\right|\leq e^{\psi(A)}\leq e^{\left|A\right|} \label{eq:exponential}
\end{align}
	\end{subequations}
In \eqref{eq:exponential}, $\psi(A)$ is the matrix defined by $[\psi(A)]_{ij}=|A_{ij}|$ if $i\neq j$ and $[\psi(A)]_{ij}=A_{ij}$ if $i=j$. 
Also, \eqref{eq:equality-norm}  means that the $2$-norms\footnote{Recall that, as defined in the notation section, the $2$-norm is here the same as the Frobenius norm.} of $A$ and $\left|A\right|$ are equal. 
\end{lem}
Indeed  $\psi(A)$ is the matrix obtained from $A$ by taking the absolute value of the off-diagonal elements and leaving entries on the main diagonal unchanged. Hence for any square real matrix $A$, $\psi(A)$ is a Metzler matrix (and if $A$ is itself a Metzler matrix, then $\psi(A)=A$).   The facts \eqref{eq:sum}-\eqref{eq:equality-norm}  which were stated in \cite[Chap. 8]{Horn85-Book} are straightforward to check. As to \eqref{eq:exponential}, it can be proved as follows.

\begin{proof}[Proof of Fact \eqref{eq:exponential}]
\textit{First inequality:} 
Let us decompose the matrix $A$ in the form $A=\diag(A)+A_0$ where $\diag(A)$ is a diagonal matrix containing the diagonal elements of $A$ and $A_0$ is a matrix having the same off-diagonal elements as $A$ and zeros on the main diagonal. Then $\psi(A)=\diag(A)+|A_0|$.  Let $\alpha\in \Re$ be such that $\alpha I+\diag(A)\geq 0$. By definition of the exponential of a  matrix, we have
$$\begin{aligned}
	\left|e^{\alpha I+A}\right|&=\left|e^{\alpha I+\diag(A)+A_0} \right|\\ 
	&\leq \sum_{k=0}^{+\infty}\dfrac{1}{k!}\left|\alpha I+\diag(A)+A_0\right|^k\\
	&\leq \sum_{k=0}^{+\infty}\dfrac{1}{k!}\big[(\alpha I+\diag(A))+|A_0|\big]^k\\
	&= \sum_{k=0}^{+\infty}\dfrac{1}{k!}\big[\alpha I+\psi(A)\big]^k
	=e^{\alpha I+\psi(A)}
\end{aligned}$$
Note that here we have used the facts \eqref{eq:sum}-\eqref{eq:product} to derive the first and second inequalities. We have hence shown that 
$\left|e^{\alpha I+A}\right|\leq e^{\alpha I+\psi(A)}$ which is equivalent to $e^{\alpha}\left|e^{A}\right|\leq e^\alpha e^{\psi(A)}$ and the result follows. 
\\
\textit{Second inequality:} We use the identity $0\leq A\leq B \: \Rightarrow A^k\leq B^k$ for any integer $k\geq 1$. It follows, by invoking the definition of matrix exponential, that $0\leq A\leq B \: \Rightarrow e^A\leq e^B$. To prove the second inequality, note that by the fact that $\psi(A)$ is Metzler, there is $\lambda\in \Re$ such that $\lambda I+\psi(A)\geq 0$. On the other hand, $\psi(A)\leq |A|$. Combining these two observations leads to $0\leq \lambda I+\psi(A)\leq \lambda I +|A|$. And by applying the identity above, we get 
$e^{\lambda I+\psi(A)}\leq e^{\lambda I +|A|} $ which implies that $e^\lambda e^{\psi(A)}\leq e^\lambda e^{|A|}$ and finally that $ e^{\psi(A)}\leq e^{|A|}$. 
\end{proof}
\noindent An alternative proof of the first inequality of \eqref{eq:exponential} can be found in \cite{Hinrichsen07} 
where $\psi(A)$ is called the Metzler part of $A$.

In order to reduce the complexity associated with the implementation of \eqref{eq:px-continuous-time}, we discuss three over-approximation methods. 
\subsubsection{Over-estimating $p_x$}
The following proposition allows to over-estimate $p_x$ with a  vector-valued signal $\hat{p}_x$ whose computation is cheaper. 
More specifically we can avoid numerical evaluation of integrals on unbounded time intervals thanks to the following proposition. 
\begin{prop}\label{prop:Truncated-Time}
Let $T\in \Re_+$. Let $\hat{p}_x:\Re_+\rightarrow \Re_+^n$ be defined by: $\hat{p}_x(t)=p_x(t)$ for all $t\in \interval[open right]{0}{T}$ where $p_x$ is defined as in \eqref{eq:px-continuous-time}, and 
\begin{equation}\label{eq:pxhat}
	\hat{p}_x(t)=|e^{AT}|\hat{p}_x(t-T)+\int_{t-T}^t|e^{A(t-\tau)}B|p_w(\tau)d\tau
\end{equation}
for all $t\geq T$ with $A$ and $B$ being the matrices of system \eqref{eq:LTI} and $p_w$ as in Theorem \ref{thm:Tightest-Estimator}. Then $p_x(t)\leq \hat{p}_x(t)\:  \forall t\in \Re_+$ and hence the state trajectories generated by system \eqref{eq:LTI} satisfy
\begin{equation}\label{eq:Enclosing-pxhat}
	c_x(t)-\hat{p}_x(t)\leq x(t)\leq c_x(t)+\hat{p}_x(t) \: \forall t\in \Re_+.
\end{equation}
with $c_x$ defined as in \eqref{eq:cx-continuous-time}. 
\end{prop}
\begin{proof}
The solution to \eqref{eq:LTI} can be written as 
$$x(t)=e^{AT}x(t-T)+\int_{t-T}^te^{A(t-\tau)}Bw(\tau)d\tau.  $$
Now, by applying Lemma \ref{lem:Tightest-Interval}, it is immediate that \eqref{eq:Enclosing-pxhat} holds if we can establish that $\mathcal{C}(c_x(t-T),\hat{p}_x(t-T))$ is an enclosing interval for $x(t-T)$. This in turn is true if  ${p}_x(t-T)\leq \hat{p}_x(t-T)$ for all $t$.  
Hence let us show that $p_x(t)\leq \hat{p}_x(t)$ for all $t$.  For this purpose, write $t$ in the form $t=q(t)T+r(t)$ where $q(t)$ is a  nonnegative integer and $r(t)\in \Re_+$ with  $0\leq r(t)< T$. Then by applying repeatedly \eqref{eq:pxhat} leads to
$$\begin{aligned}
	\hat{p}_x&(t)=|e^{AT}|^{q(t)} \hat{p}_x(r(t))+\\
	&+\sum_{j=1}^{q(t)}\int_{t-jT}^{t-(j-1)T}\big|e^{AT}\big|^{(j-1)}\big| e^{A(t-(j-1)T-\tau)}B\big|p_w(\tau)d\tau
\end{aligned}$$
By applying  \eqref{eq:product}, we see that 
$$\begin{aligned}
	\hat{p}_x(t)\geq &|e^{ATq(t)}| \hat{p}_x(r(t))+\int_{t-q(t)T}^{t}\big| e^{A(t-\tau)}B\big|p_w(\tau)d\tau
\end{aligned}$$
On the other hand $\hat{p}_x(r(t))={p}_x(r(t))=|e^{Ar(t)}|p_x(0)+\int_{0}^{r(t)}\big|e^{A(t-\tau)}B\big|p_w(\tau)d\tau$. Plugging this in the last inequality above and applying again \eqref{eq:product} show that $\hat{p}_x(t)\geq p_x(t)$. 
\end{proof}
Note that if $A$ is Hurwitz, then $T$ can be chosen sufficiently large so that $|e^{AT}|$ is Schur stable\footnote{i.e., its spectral radius is less than $1$. }. For such a $T$,  $(c_x,\hat{p}_x)$ defines an interval estimator for system \eqref{eq:LTI} in the sense of Definition \ref{def:Interval-Estimator}. As shown by Proposition \ref{prop:Truncated-Time}, the interval estimate defined by $(c_x,\hat{p}_x)$ is only an over-estimate of the one resulting from $(c_x,p_x)$.  As $T$ gets larger, the two interval estimators will get closer but then the complexity increases. And in the extreme case where $T=t$, we recover $\hat{p}_x=p_x$. 

\subsubsection{Approximation using a Metzler matrix}\label{subsubsec:Metzler}
\noindent A second simple approximation can be obtained directly from Lemma \ref{lem:Inequalities}. 
In effect, by applying the facts \eqref{eq:sum}-\eqref{eq:exponential} above, we can write $p_x(t)\leq \check{p}_x(t)$ where 
\begin{equation}\label{eq:phat}
	\check{p}_x(t)\triangleq e^{\psi(A)t}p_x(0)+\int_{0}^te^{\psi(A)(t-\tau)}|B|p_w(\tau)d\tau.
\end{equation}
Although this is a looser estimate of $p_x$ (than e.g., \eqref{eq:pxhat}) its benefit lies in the fact that it is easier to compute. In effect, the new signal $\check{p}_x$ can be realized very simply in the form
$\dot{\check{p}}_x(t)=\psi(A)\check{p}_x(t)+|B|p_w(t)$ with $\check{p}_x(0)=p_x(0)$. However for $(c_x,\check{p}_x)$ to be an interval estimator in the sense of Definition \ref{def:Interval-Estimator}, we must require additionally that $\psi(A)$ be Hurwitz.

\subsubsection{Over-estimating $p_w$ by a constant vector}\label{subsubsec:pw}
Another over-estimate of $p_x$ can be obtained from Proposition \ref{eq:px-for-constant-pw} as follows. By Assumption \ref{assum:Bounding}, $p_w$ is bounded. Therefore, let $\delta^o$ be the vector in $\Re^{n_w}$ whose $i$-th   entry $\delta_i^o$ is defined by $\delta_i^o=\sup_{t\in \Re_+}p_{w,i}(t)$ where $p_{w,i}(t)$ refers to the $i$-th entry of $p_w(t)$. Then by letting $\delta$ be a signal defined by  $\delta(t)=\delta^o$ for all $t\geq 0$,  $w$ satisfies $c_w(t)-\delta(t)\leq w(t)\leq c_w(t)+\delta(t)$ and hence $(c_w,\delta)$ is a valid interval representation for the input signal $w$ which fulfills the condition of Proposition \ref{eq:px-for-constant-pw}.  As a consequence, replacing $p_w(0)$ in \eqref{eq:realization-px} with $\delta^o$ gives a computable realization of an interval estimator for the state of system \eqref{eq:LTI}. \\
For an empirical comparison of the estimators  discussed here, see Section \ref{sec:Simulations}. 

\section{Closed-loop state estimator for LTI systems}\label{sec:Closed-Loop}
In case the system \eqref{eq:LTI} is not stable, let us assume it to be observable (or just detectable). Then it is possible to find a matrix gain $L$ such that $A-LC$ is Hurwitz. We can then construct an interval observer from the classical observer form. As we did in open-loop, we can of course write the best estimator \eqref{eq:IntEst}-\eqref{eq:cx-continuous-time} also in closed-loop for a given $L$ or compute its over-approximations discussed in Section \ref{subsec:Approximation}. However here we choose to study further the type of approximation given in \eqref{eq:phat}. Although this type of estimator is not the tightest one,  it has the advantage of computational simplicity.

\subsection{A systematic design method}
In this section we discuss a systematic way of constructing interval observers  employing an output injection.
Departing from the structure of the classical Luenberger observer, it is   easy to see that the  state of system \eqref{eq:LTI} satisfies 
\begin{equation}\label{eq:observer}
	\dot{x}(t) =(A-LC)x(t)+Gs(t),
\end{equation} 
where  $$G=\bbm B & L & -L\eem \mbox{ and } s(t) = \bbm  w(t)^\top  & y(t)^\top & v(t)^\top\eem^\top$$ 
with $L$ being the gain of the observer. 
Then by relying on the discussion of Section \ref{subsubsec:Metzler}, we can construct an enclosing interval estimate $(c_x^{\mbox{\tiny CL}},p_x^{\mbox{\tiny CL}})$ for the state of system \eqref{eq:LTI} by
\begin{equation}
	\label{eq:Interval-Luenberger}
\hspace{-7pt}\left\{	\begin{aligned}
	& \dot{c}_x^{\mbox{\tiny CL}}=(A-LC)c_x^{\mbox{\tiny CL}}(t)+Gc_s(t), \: c_x^{\mbox{\tiny CL}}(0)=c_x(0)\\
	&\dot{p}_x^{\mbox{\tiny CL}}=\psi(A-LC){p}_x^{\mbox{\tiny CL}}(t)+\left|G\right|p_s(t), \: p_x^{\mbox{\tiny CL}}(0)=p_x(0),
	\end{aligned}\right. 
\end{equation}
where $(c_s(t), p_s(t))\in \Re^{n_s}\times\Re_+^{n_s}$, $n_s=n_w+2n_y$, is a center-radius representation of $s(t)$.  
The systems \eqref{eq:Interval-Luenberger}  yield an interval observer for system \eqref{eq:LTI} provided that both $A-LC$  and $\psi(A-LC)$ are Hurwitz. 
By the statement (a) of Lemma \ref{lem:Stability} stated below, this stability condition is satisfied if and only if $\psi\left(A-LC\right)$ is Hurwitz. 
\begin{lem}\label{lem:Stability}
Let $A, A_1, A_2\in \Re^{n\times n}$ and $\mathcal{P}\in \Re_+^{n\times n}$. Let $\psi$ be the matrix  function defined in Lemma  \ref{lem:Inequalities}. Then the following implications hold: 
\begin{itemize}
	\item[(a)] $\psi(A)$ is Hurwitz $\Rightarrow$ $A$ is Hurwitz.
	\item[(b)]  $\psi(A)+\mathcal{P}$ is Hurwitz $\Rightarrow$ $\psi(A)$ is Hurwitz
	\item[(c)] $\psi(A_1)\leq \psi(A_2)$ $\Rightarrow$ $0\leq e^{\psi(A_1)}\leq e^{\psi(A_2)}$
	\item[(d)] If $\psi(A_1)\leq \psi(A_2)$, then  $\psi(A_1)$ is Hurwitz whenever $\psi(A_2)$ is Hurwitz 
\end{itemize}
Here, all matrix inequalities with symbol $\leq$ are understood entrywise. 
\end{lem}
\begin{proof}\textit{Proof of (a):}
As a starting point of the proof of item (i), note that a matrix $X\in \Re^{n\times n}$ is Hurwitz if and only if there exist some positive constants $c$ and $\lambda$ such that $\|e^{Xt}\|_2<ce^{-\lambda t}$ for all $t\in \Re_+$. By applying \eqref{eq:inequality-norm}-\eqref{eq:exponential}, we can write 
$$  \left\|e^{At}\right\|_2\leq \big\|e^{\psi(A)t}\big\|_2$$ 
for all $t\geq 0$. 
Hence, if  $\psi(A)$ is Hurwitz then so is $A$. \\
\textit{Proof of (b): } Since $\psi(A)$ is a Metzler matrix, we can find a number $\alpha\geq 0$ such that $0\leq\alpha I+\psi(A)\leq \alpha I+\psi(A)+\mathcal{P}$.  It follows that 
$0\leq e^{\alpha I+\psi(A)}\leq e^{\alpha I+\psi(A)+\mathcal{P}}$ and hence that $0\leq e^{\psi(A)}\leq e^{\psi(A)+\mathcal{P}}$. Applying property \eqref{eq:equality-norm} then  yields $\big\|e^{\psi(A)}\big\|_2\leq \big\|e^{\psi(A)+\mathcal{P}}\big\|_2$ from which the conclusion follows. \\
\textit{Proofs of (c) and (d): } (c) follows by a similar reasoning as in the proof of item (b). 
As to the statement (d), it follows from (c) and \eqref{eq:equality-norm}. 
\end{proof}

\noindent The question now is how to effectively select a matrix gain $L\in \Re^{n\times n_y}$ so as to realize the condition $\psi\left(A-LC\right)$ is Hurwitz. An answer is provided by the following lemma. 
\begin{lem}\label{lem:LMI}
Let $(A,C)\in \Re^{n\times n}\times \Re^{n_y\times n}$. Then the following statements are equivalent:
\begin{enumerate}
	\item[(e)] There exists $L\in \Re^{n\times n_y}$  such that $\psi(A-LC)$ is Hurwitz.   
\item[(f)] There exist a \textit{diagonal positive definite} matrix $P\in \Re^{n\times n}$  and some matrices $Y\in \Re^{n\times n_y}$, $X\in \Re^{n\times n}$ satisfying the conditions:
\begin{equation}\label{eq:LMI}
	\begin{aligned}
		&X^\top +X+2\diag(S)\prec 0 \\
		& |S-\diag(S)|\leq X
	\end{aligned}
\end{equation}
\end{enumerate}
where $S = PA-YC$. 
In case the statements hold, $L$ is given by $L=P^{-1}Y$. 
\end{lem}
\begin{proof}
Let $\bar{A}=A-LC$.  Since $\psi(\bar{A})$ is a Metzler matrix, we can apply Theorem 15 in \cite[p. 41]{Farina00-Book} to state that $\psi\left(\bar{A}\right)$ is Hurwitz if and only if 
there exists a diagonal and positive definite matrix $P$ such that
$$\psi(\bar{A})^\top P+ P\psi(\bar{A})\prec 0. $$
Observe that $\psi(\bar{A})=\diag(\bar{A})+\left|\bar{A}-\diag(\bar{A})\right|$ with $\diag(\bar{A})$ being the matrix formed with the diagonal entries of $\bar{A}$. Moreover since $P$ is a diagonal matrix with strictly positive entries, $P|S|=|PS|$ for any $S\in\Re^{n\times n}$.  Using these two remarks we can write the above stability condition as 
$$\begin{aligned}
	2P\diag(\bar{A})&+|PA-YC-P\diag(\bar{A})|^\top \\
	&\qquad + |PA-YC-P\diag(\bar{A})|\prec 0,
\end{aligned}  $$
where $Y=PL$. 
By letting now $S=PA-YC$ and noting that $P\diag(\bar{A})=\diag(S)$,  
we get that condition (e) in the lemma is equivalent to 
\begin{equation}\label{eq:equi-(e)}
	2\diag(S)+|S-\diag(S)|^\top + |S-\diag(S)|\prec 0.
\end{equation}
(i.e., condition (e) is equivalent to $\psi(S)+\psi(S^\top)\prec 0$). 
Hence by setting $X= |S-\diag(S)|$ we see that condition (e) implies condition (f).  \\ 
Now assume that the condition (f)  is satisfied. Then $\psi(S)-\diag(S)=|S-\diag(S)|\leq X$ and hence 
$$\psi(S)+\psi(S^\top)\leq X+X^\top +2\diag(S).$$
 Note that the terms on both sides of this inequality are symmetric and Metzler. On the other hand, the first inequality of \eqref{eq:LMI} implies  that $X+X^\top +2\diag(S)$ is Hurwitz. Therefore we can apply the statement  (d) of Lemma \ref{lem:Stability} to conclude that $\psi(S)+\psi(S^\top)$ is also Hurwitz. Now by the fact that (e) is equivalent to $\psi(S)+\psi(S^\top)\prec 0$  (see \eqref{eq:equi-(e)}) we find that (f) implies (e). 
 This concludes the proof. 
\end{proof}
\noindent Lemma \ref{lem:LMI} shows that one can compute the observer gain $L$ efficiently by solving a feasibility problem which is expressible in terms of Linear Matrix Inequalities (LMI) \cite{Boyd97-LMI-Book}. In comparison to classical results we do not require $A-LC$ to be Metzler since $\psi(A-LC)$ is naturally Metzler. Hence the only constraint associated with the search for the gain $L$ is the Hurwitz stability of $\psi(A-LC)$.

\begin{rem}
In addition to ensuring stability, the gain $L$ could be designed so as to guarantee a certain level of convergence speed. For that it suffices to replace the first equation of \eqref{eq:LMI} with $X^\top +X+2\diag(S)\prec -\alpha P$ with $\alpha>0$ a predefined level of decay and $P$ being the diagonal positive definite matrix of Lemma \ref{lem:LMI}. Also it can be of interest, similarly as in \cite{Rami08-CDC}, to select the matrix $L$ so that to minimize a performance index of the form 
$$\int_{0}^\infty \phi\left( p_x^{\mbox{\tiny CL}}(\tau)\right) d\tau $$ 
subject to the condition of Lemma \ref{lem:LMI}, with $\phi$ being some cost function.  
\end{rem}
\vspace{-.3cm}
\section{Numerical results}\label{sec:Simulations}
This section reports some simulation results that illustrate the performances of some of the interval estimators discussed in this paper. For concision, we just consider the open-loop configuration. 
Consider an instance of system \eqref{eq:LTI} with fixed-values state transition matrices  defined by 
\begin{equation}\label{eq:example}
	A=\BM -3 & 1.5\\ -2  & -2\EM \: \:  \mbox{ and }\: \: B=\BM -1\\ 0\EM.
\end{equation}
The input $w$ is such that $w(t)\in \mathcal{C}\big(c_w(t),p_w(t)\big)$ for all $t$ where
$c_w(t)=5\sin(2\pi\nu_c t)$ and $p_w(t)=\left|2\sin(2\pi\nu_p t)\right|$ with $\nu_c = 0.3$  and $\nu_p=50$. 
As to the initial state, it lives in an interval $\mathcal{C}\big(c_x(0),p_x(0)\big)$ with 
$c_x(0)=\bbm -2 & 2\eem^\top$, $p_x(0)=\bbm 3 & 2.2\eem^\top$. Note that in order to be able to test all the estimators in open-loop (in particular the one suggested in Section \ref{subsubsec:Metzler}), the matrix $A$ in \eqref{eq:example} has been selected such that $\psi(A)$ is Hurwitz. 

For this example, Figure \ref{fig:state} compares the tightest estimator proposed in \eqref{eq:IntEst}-\eqref{eq:px-continuous-time} with three estimators from the family described in Eqs \eqref{eq:pxhat}-\eqref{eq:Enclosing-pxhat} for $T\in \left\{0.01, 0.1, 1\right\}$. Two comments can be made. First, these simulation results provide an empirical evidence supporting our claim that the estimator proposed in \eqref{eq:IntEst}-\eqref{eq:px-continuous-time} is indeed the tightest possible. Second, the over-approximation given in \eqref{eq:pxhat}-\eqref{eq:Enclosing-pxhat} gets tighter as the horizon $T$ increases. Finally, it is interesting to observe that $T$ needs not be too large for $\hat{p}_x$ in \eqref{eq:pxhat} to provide a good approximation of $p_x$; here we get a good match between $p_x$ and $\hat{p}_x$ for a value as small as $T=1$.  

The second figure (Fig. \ref{fig:state2}) compares the estimator \eqref{eq:IntEst}-\eqref{eq:px-continuous-time} to its over-approximations discussed in Sections \ref{subsubsec:Metzler} and \ref{subsubsec:pw}. A specificity of these estimators is that they are computationally less expensive to implement as they can be realized by finite dimensional  state-space representations (with state lengths equal to $2n$ and $n(n+1)$ respectively). It follows from the empirical results that in the current settings, the over-approximation using the Metzler matrix $\psi(A)$ is the cheapest but also the least tight. 
 %
\begin{figure}
	\centering
	\psfrag{t}[][]{\tiny time [s]}
	\psfrag{x}[][]{\scriptsize $x_1$} %
	\psfrag{Tsim}{\tiny Tightest} 
	\psfrag{T1}{\tiny $T=1$} 
	\psfrag{T01}{\tiny $T=0.1$}
	\psfrag{T001}{\tiny $T=0.01$} 
\includegraphics[width=8cm,height=5cm]{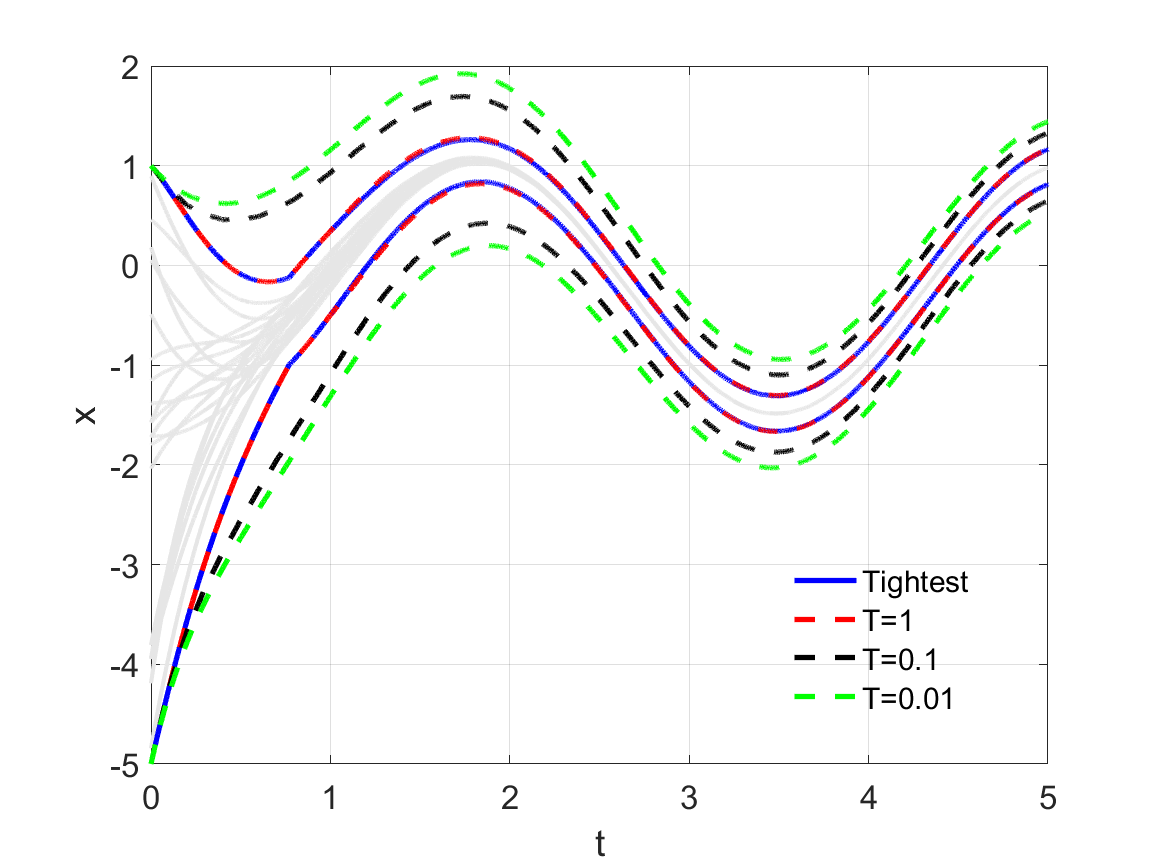}
\psfrag{x}[][]{\scriptsize $x_2$}
\includegraphics[width=8cm,height=5cm]{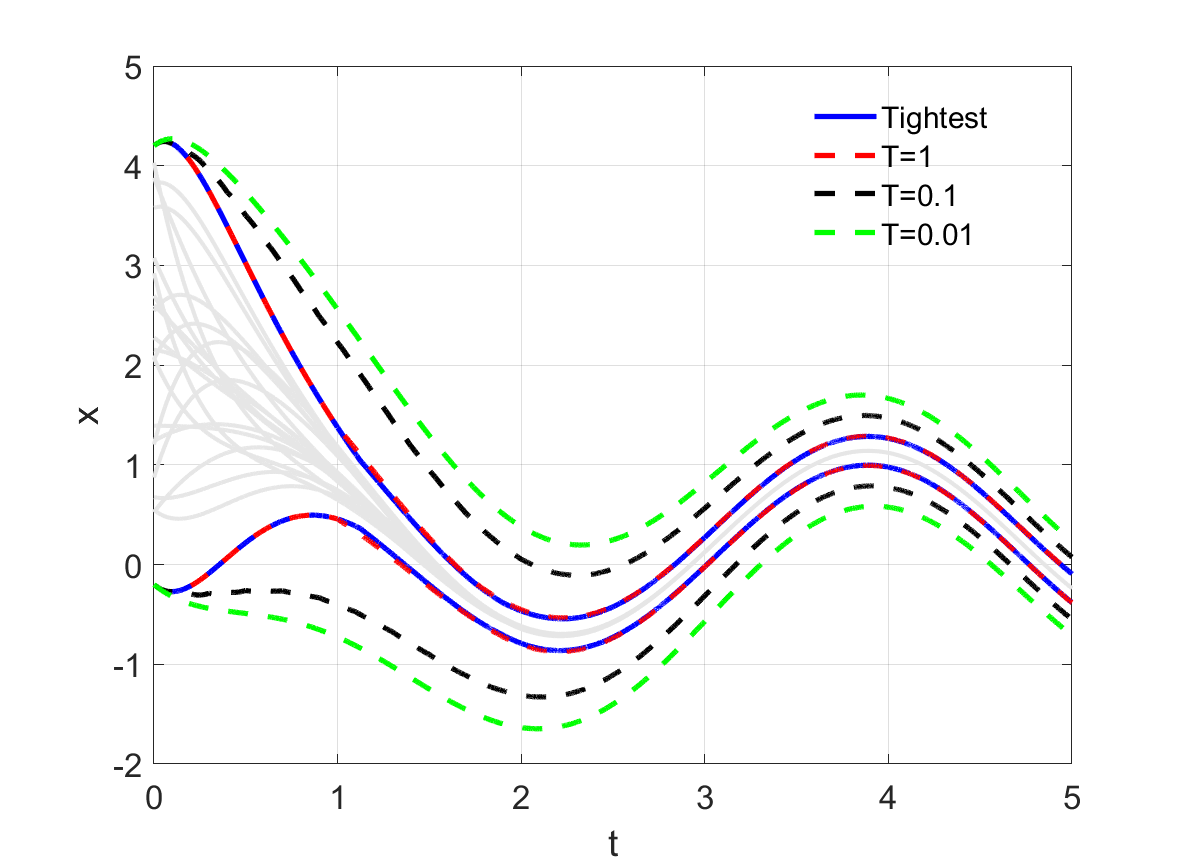}
\caption{Comparison of open-loop interval estimators \eqref{eq:pxhat}-\eqref{eq:Enclosing-pxhat} for $T=0.01$ (green), $T=0.1$ (black), $T=1$ (red), $T=t$ (blue). In gray are represented the state trajectories of the system generated from different initial conditions and different inputs with values on the allowed intervals.  }
	\label{fig:state}
\end{figure}
\begin{figure}
	\centering
	\psfrag{t}[][]{\tiny time [s]}
	\psfrag{x}[][]{\scriptsize $x_1$} %
	\psfrag{Tsim}{\tiny Tightest} 
	\psfrag{T1}{\tiny $T=1$} 
	\psfrag{T01}{\tiny $T=0.1$}
	\psfrag{T001}{\tiny $T=0.01$} 
\includegraphics[width=8cm,height=5cm]{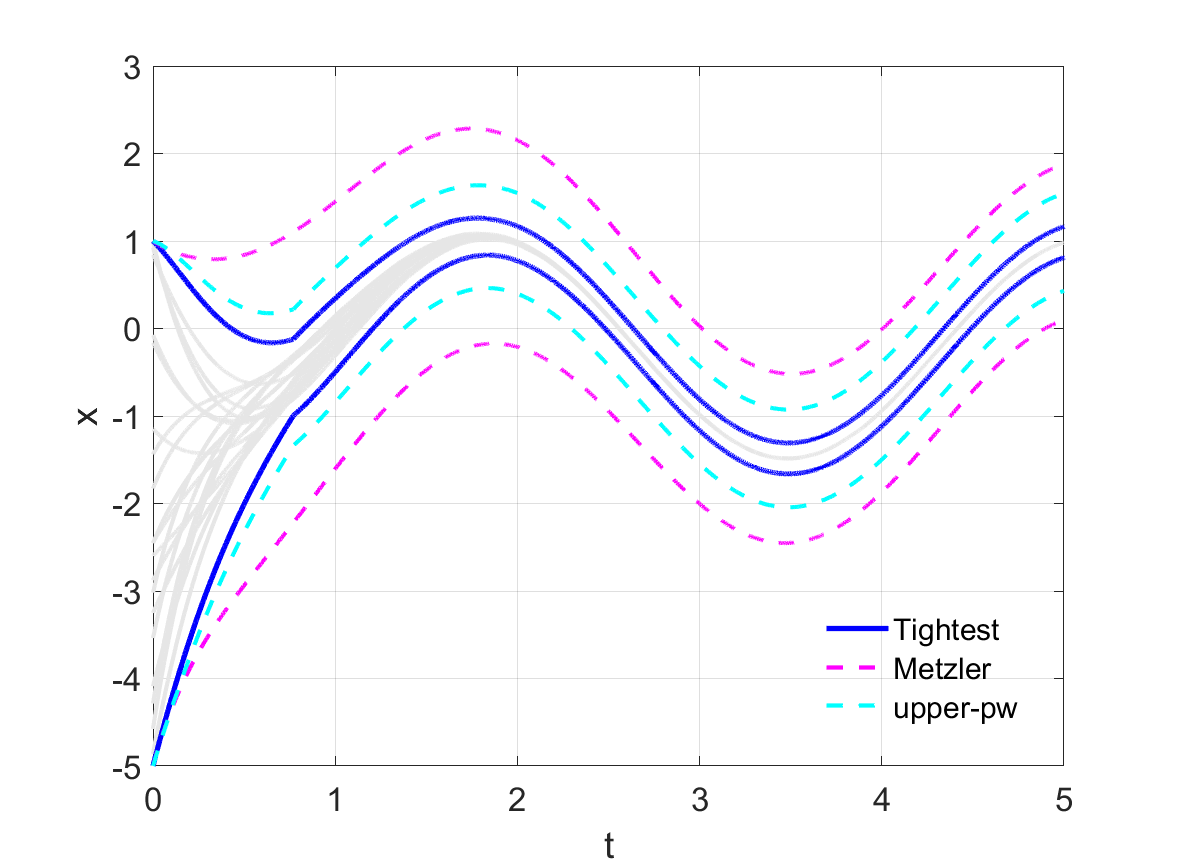}
\psfrag{x}[][]{\scriptsize $x_2$}
\includegraphics[width=8cm,height=5cm]{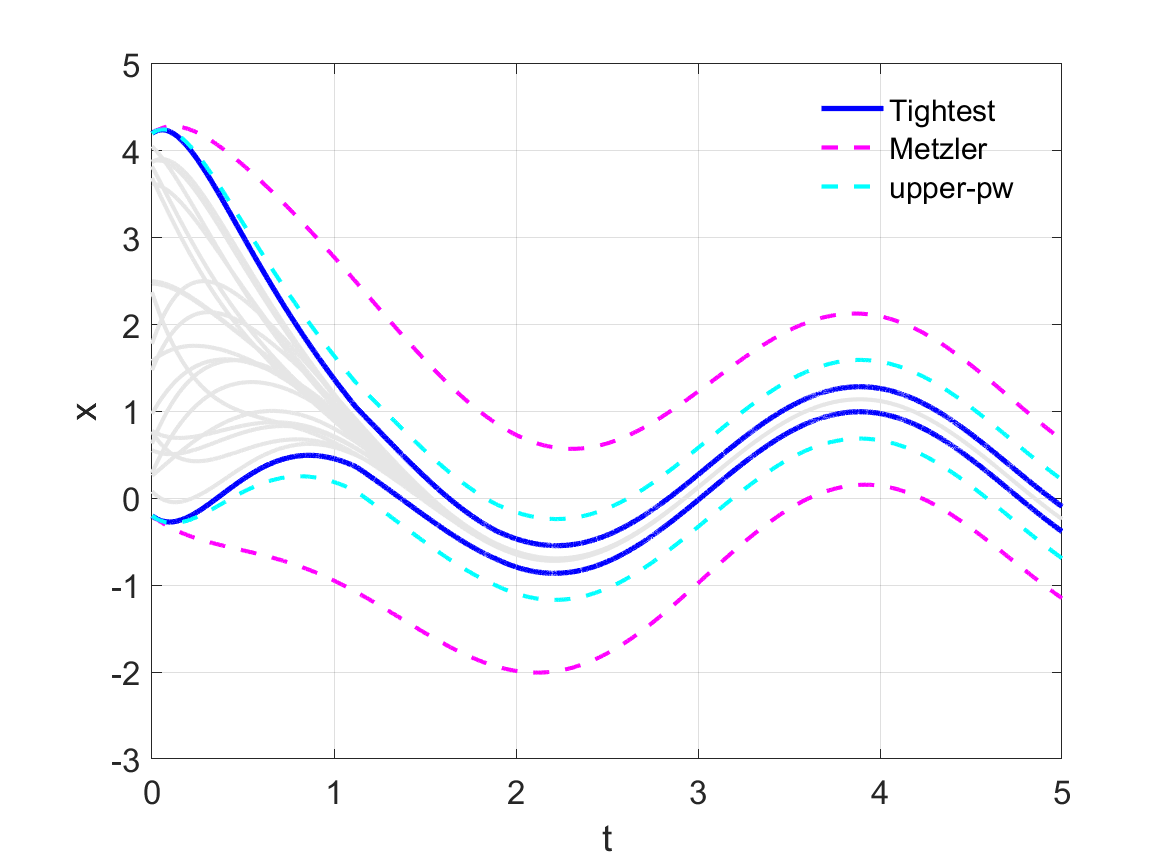}
\caption{Comparison of open-loop interval estimators: tightest (blue), estimator \eqref{eq:realization-px} (cyan) obtained by upper-bounding $p_w$ with $2$, approximation using a Metzler matrix (magenta). In gray are represented the state trajectories of the system generated from different initial conditions and different inputs.  }
	\label{fig:state2}
\end{figure}

\section{Conclusion}\label{sec:Conclusion}
In this paper we have presented a new approach to the interval-valued state estimation problem. The proposed framework is mainly discussed for the case of continuous-time linear systems but it is  generalizable (to some extent) to  LTV systems and probably to some other  classes of systems.  The main contribution of this work consists in the derivation of the tightest interval-valued estimator which encloses all the possible state trajectories generated by an uncertain LTI system. A  numerical implementation of this estimator requires however some trade-off between tightness and computational load. Therefore some relaxations on tightness have been  discussed.  


\bibliographystyle{abbrv}

\end{document}